\documentclass[letterpaper, 10 pt, conference]{ieeeconf} 

\usepackage{amsmath}
\usepackage{amssymb}
\usepackage{dsfont}
\usepackage{bbold}
\usepackage{graphicx}
\usepackage{tikz}
\usetikzlibrary{arrows}

\usepackage{thmtools}
\usepackage{answers}
\usepackage{lipsum}

\renewcommand\thmcontinues[1]{Continued}

\newtheorem{theorem}{Theorem}
\newtheorem{proposition}{Proposition}
\newtheorem{definition}{Definition}

\newtheorem{remark}{Remark}
\newtheorem{example}{Example}

\DeclareMathOperator{\dist}{dist}

\DeclareMathOperator*{\argmin}{arg\,min}
\newcommand{\mc}{\mathcal}
\newcommand{\mb}{\mathbf}

\newcommand{\dd}{\mathrm{d}}
\newcommand{\oo}{\mathrm{o}}

\newcommand{\R}{\mathds{R}}

\newcommand{\wz}{\tilde{z}}

\IEEEoverridecommandlockouts                             

\title{\LARGE \bf
Stability and bifurcations in transportation networks with heterogeneous users}

\author{Leonardo Cianfanelli, Giacomo Como and Tommaso Toso, \IEEEmembership{Member, IEEE} \thanks{Leonardo Cianfanelli and Giacomo Como are with Department of Mathematical Sciences G.L.~Lagrange, Politecnico di Torino, Corso Duca degli Abruzzi 24, 10129, Torino, Italy.  Email: leonardo.cianfanelli@polito.it, giacomo.como@polito.it. Tommaso Toso is with Univ. Grenoble Alpes, CNRS, Inria, Grenoble INP, GIPSA-lab, 38000 Grenoble, France. Email: tommaso.toso@gipsa-lab.fr.}
\thanks{}
}

\begin{document}

\maketitle
\thispagestyle{empty}
\pagestyle{empty}

\begin{abstract}
A critical aspect in strategic modeling of transportation systems is user heterogeneity. In many real-world scenarios, e.g., when tolls are charged and drivers have different trade-offs between time and money, or when they get informed about current congestion by different routing apps, modeling users as rational decision makers with homogeneous utility functions becomes too restrictive. While global asymptotic stability of user equilibria in homogeneous routing games is known to hold for a broad class of evolutionary dynamics, the stability analysis of user equilibria in heterogeneous routing games is a largely open problem. In this work we study the logit dynamics in heterogeneous routing games on arbitrary network topologies. We show that the dynamics may exhibit bifurcations as the noise level of the dynamics varies, and provide sufficient conditions for asymptotic stability of user equilibria.
\end{abstract}

\section{Introduction}
Due to the rising congestion level of urban areas and the fast-increasing pervasiveness of novel intelligent technology that is having a huge impact on the transportation system, the analysis, design, and control of traffic networks have received renewed attention. A key aspect to be properly addressed in this research is the fact that routing apps and other information technology systems are completely reshaping users' behaviour. Given the increasing amount of available information and the selfish and often competing objectives of the users, it is natural to incorporate game-theoretic aspects in traffic models.

An important aspect in game-theoretic traffic models is concerned with the user preferences. The most popular model assumes \emph{homogeneity}, i.e., that all users make decisions based on identical utility functions, given their available information \cite{wardrop1952road}. However, this assumption may prove too restrictive to model many real-world scenarios of interest, e.g., when drivers use different routing apps \cite{wu2017informational,wu2020value}, when fuel consumption or monetary tolls constitute a non-negligible fraction of the cost and users have different trade-offs between time and money \cite{cole2003pricing,fleischer2004tolls}, or when users have different knowledge on the available routes \cite{acemoglu2018informational}. Homogeneous models have been first generalized in \cite{dafermos1972traffic} to account for heterogeneity of the utility functions.
From now on, we shall refer to \emph{heterogeneous routing games} to denote game-theoretic models that incorporate user heterogeneity, in contrast with \emph{homogeneous routing games}, which do not consider this aspect.

Besides user heterogeneity, another crucial aspect in game-theoretic models is the evolution of network flows under evolutionary dynamics, which describe how users revise their strategies. The distinction between homogeneous and heterogeneous routing games has several implications on the properties of the game, in particular on the stability of the user equilibria under evolutionary dynamics. While global asymptotic stability of user equilibria in homogeneous routing games is known to hold for a broad class of evolutionary dynamics \cite{sandholm2010population}, %(including the logit dynamics), 
their stability in heterogeneous routing games is a largely open issue. Besides the theoretical interest, stability of equilibria has practical implications and paves the way for control applications. Indeed, since heterogeneous routing games may admit multiple user equilibria \cite{milchtaich2005topological}, understanding whether the network flows will converge to an equilibrium, and identifying which one will be selected by the dynamics in case of non-uniqueness, are fundamental questions for a system planner that aims at optimizing the transportation network performance.

In most of the literature dealing with user heterogeneity, a big effort is spent to analyse the static properties of the equilibria, but the stability of such equilibria is typically not investigated \cite{wu2017informational,wu2020value,cole2003pricing, thai2016negative,wu2019information}, implicitly assuming that the network flows converge to the equilibria of the game. However, this assumption is not always justified and requires to be further motivated. To the best of our knowledge, the only stability result in heterogeneous routing games states that a sufficient condition for global asymptotic stability of the equilibria is that the graph has parallel routes, or it is the series composition of graphs with parallel routes \cite{cianfanelli2019stability}. %No results on arbitrary graph topologies are found in the literature. \\ %The speed of convergence of no-regrets dynamics and imitative dynamics are analysed in \cite{blum2006routing} and \cite{fischer2004evolution}, but in the considered games the populations differ only in the origin-destination pair and not in the delay functions. 
%In the monograph from Sandholm, results on the global stability of equilibria in stable, potential or supermodular games are provided, which however do not include the case of heterogeneous routing games \cite{sandholm2010population}.\\

In this work we investigate the behaviour of the logit dynamics, which models users that aim at choosing optimal routes, but due to imperfect information or incomplete rationality may sometimes select suboptimal ones. We establish novel results that hold for every heterogeneous routing game, independently of the network topology. %We characterize the fixed points of the dynamics for arbitrary noise values, and then focus on the vanishing noise and large noise limits of the logit dynamics. 
Our contribution is the following. We first characterize the set of fixed points of the logit dynamics (the expression \emph{fixed points} is used in this context to avoid any source of confusion between equilibrium points of the logit dynamics and user equilibria of the routing game), and prove that such a set approaches a subset of the Wardrop equilibria of the game (called \emph{limit equilibria}) in the vanishing noise limit. We then show that all the strict equilibria of the game (i.e., equilibrium flows under which every population uses a single route and all the other routes are strictly suboptimal) belong to the set of limit equilibria, and prove their local asymptotic stability under the logit dynamics. We also show that, in the large noise limit, the logit dynamics admits a globally asymptotically stable equilibrium in every heterogeneous routing games. Finally, we conduct numerical simulations to validate our theoretical results, and show that the dynamics may exhibit bifurcations as the noise varies, in contrast with the behaviour observed in homogeneous routing games.

%Evolutionary game theory has been first formulated in \cite{smith1982evolution} to describe animal behaviour in game-theoretic situations, and then applied more generally to the evolution of strategic choices in game theory \cite{hofbauer2003evolutionary}. For a complete reference on evolutionary dynamics in population games we refer to \cite{sandholm2010population}. 
%To the best of our knowledge, no theoretical results on the global stability under evolutionary dynamics of Wardrop equilibria of heterogeneous routing games are provided in the literature. The speed of convergence to Wardrop equilibria in homogeneous routing games is studied in \cite{blum2006routing} for no-regrets dynamics, and in \cite{fischer2004evolution} for imitative dynamics, but heterogeneous routing games are not included in the analysis. In \cite{kleinberg2009multiplicative} the convergence of evolutionary dynamics in homogeneous atomic games is investigated. 
%In \cite{como2013stability,como2021distributed}, the authors propose a multiscale model in which the dynamics of users' choice are intertwined with the physical dynamics on the network, but the users are assumed homogeneous.

The rest of the paper is organized as follows. In Section \ref{sec:model} we define heterogeneous routing games, introduce the logit dynamics, and discuss a motivating example. In Section \ref{sec:logit_routing} we establish our novel results on the logit dynamics in heterogeneous routing games. Finally, in the conclusive Section \ref{sec:conclusion}, we summarize the results and discuss future research lines.

\subsection*{Notation}
Let $\mathds{R}$ and $\mathds{R}_+$ denote the set of real numbers and non-negative reals. Given a finite set $\mc{X}$, we let $\mathds{R}^{\mc{X}}$ the space of real-valued vectors whose elements are indexed by $\mc{X}$, and $|\mc{X}|$ denote the cardinality of $\mc{X}$. Let $\delta^{(i)}$, $\mb{1}$, $\mb{0}$, and $\mb{I}$ denote the vector with $1$ in position $i$ and $0$ in all the other positions, the vector of all ones, the matrix of all zeros, and the identity matrix, respectively, where the size may be deduced from the context.
The distance between a point $y$ in $\R^n$ and a set $\mc X\subseteq\R^n$ is defined as 
$$\dist(y,\mc X)=\inf\{||y-x||:\,x\in\mc X\}\,,$$
%while, for $\eps>0$, 
%$$\mc B_{\eps}(\mc X)=\{y\in\R^n:\,\dist(y,\mc X)<\eps\}$$
%denotes the set of points at distance below $\eps$ from $\mc X$. 

\section{Model}
\label{sec:model}
In this section we define the model and discuss a motivating example. Specifically, in Section \ref{sec:heterogeneous} we describe heterogeneous routing games. Then, in Section \ref{sec:logit}, we introduce the logit dynamics and provide numerical simulations of the dynamics in a heterogeneous routing game.
\subsection{Heterogeneous routing games}
\label{sec:heterogeneous}
We model the transportation network as a directed multigraph $\mathcal{G}=(\mathcal{N},\mathcal{E})$, with node set $\mc N$ and link set $\mc E$.
We consider a finite set $\mathcal{P}$ of users populations. Let each population $p$ in $\mc P$ have an origin-destination pair $(\oo_p,\dd_p)$ in $\mc N \times \mc N$ and let $v_p \ge 0$ denote the throughput of population $p$. We then stack all throughput values in a vector $v \in \mathds{R}_+^{\mc{P}}$.
Let $\mc{R}_p$ denote the set of routes from $\oo_p$ to $\dd_p$,
\begin{equation*}
	\mc{Z}_p = \{z^p \in \mathds{R}^{\mc{R}_p}_+: \mb{1}'z^p = v_p\}
\end{equation*}
indicate the set of the admissible route flows for population $p$, and $\mc{Z}$ denote the product of such sets.
%We define the \emph{aggregate flow} as
%\begin{equation}
%	\label{f_agg}
%	w=\sum_{p \in \mathcal{P}}z^p.
%\end{equation}
Every route flow $z \in \mc{Z}$ induces a unique link flow via
\begin{equation}
	\label{incidence}
	f=\sum_{p \in \mc{P}} A^p z^p,
\end{equation}
where $A^p \in \mathds{R}^{\mc{E} \times \mc{R}_p}$ is the link-route incidence matrix, with entries $A^p_{er}=1$ if the link $e$ belongs to the route $r$, or $0$ otherwise.
The populations differ in the origin-destination pair and in the delay functions according to which they make decisions. Let $\tau_e^p :\mathds{R}_+ \rightarrow \mathds{R}_+ $ denote the \emph{delay function} of link $e \in \mathcal{E}$ for population $p \in \mathcal{P}$, which is assumed a non-decreasing function of $f_e$ to take into account congestion effects. We also assume that $\tau_e^p \in \mc{C}^1$.
%This assumption looks reasonable for all the applications whereby the heterogeneity models different users' information (e.g., in case of different routing apps) or users' behaviour (e.g., different risk-adversion level in case of uncertain state of the road). If instead the heterogeneity models different types of vehicles, then the populations may contribute with different weight to the delay functions. This model can still handle this type of applications, by rescaling the flow (and the throughput) of the populations with positive parameters $\alpha^p$ proportionally to their contribution to the congestion of the links.
The cost of a route $r$ is defined as the sum of the delay functions of the links belonging to the route, i.e.,
\begin{equation}
	\label{cost}
	c_r^p(z)=\sum_{e \in \mathcal{E}} A_{er}^p \tau_e^p(f_e),
\end{equation}
where, given $z \in \mc{Z}$, the link flow $f$ is computed via \eqref{incidence}.

\begin{definition}
	\label{def:wardrop}
	A heterogeneous routing game is a quadruple ($\mathcal{G}$, $\mathcal{P}$, $\tau,v$), where $\tau$ is the vector collecting the delay functions of every link and population.
\end{definition}

We assume that the users behave as players in a game-theoretic setting, taking route with minimal cost. This behaviour is captured by the notion of Wardrop equilibrium.

\begin{definition}[Wardrop equilibrium] A Wardrop equilibrium is an admissible route flow $z \in \mc{Z}$ such that for every population $p \in \mathcal{P}$ and route $r \in \mathcal{R}_p$
	\begin{equation}
		\label{wardrop}
		z^p_r > 0 \ \Rightarrow \ c^p_r(z) \le c^p_q(z) \quad  \forall q \in \mathcal{R}_p.
	\end{equation}
A Wardrop equilibrium $z$ is called \emph{strict} if $z^p=v_p \delta^{(r)}$ for a route $r \in \mc{R}_p$ and $c^p_r(z) < c^p_s(z)$ for every $s \in \mc{R}_p \setminus \{r\}$.
\end{definition}

In other words, under Wardrop equilibrium flow, no user can unilaterally decrease her cost by changing route, because every used route by a population is optimal for that population. An equilibrium is called strict if every population uses one route only and the other routes are strictly suboptimal. %We shall see in the next sections that those equilibria have interesting properties in terms of stability. 
We let $\mc{Z}^*$ and $\mc{Z}_s^*$ denote the set of Wardrop equilibria and strict Wardrop equilibria of a routing game, respectively. 
%The next proposition provides a characterization of $\mc{Z}^*$.
%\begin{proposition}[\textcolor{red}{serve? reference?}]
%	\label{prp:existence}
%	The set of the Wardrop equilibria $\mc{Z}^*$ is non-empty and compact for every heterogeneous routing game.
%\end{proposition}
%\begin{proof}
	It is proved in \cite[Theorem 2.1.1]{sandholm2010population} that $\mc Z^*$ is never empty, i.e., there exists at least a Wardrop equilibrium. Moreover, standard arguments allow to state that $\mc Z^*$ is also compact.
%	 by using fixed point techniques. Since $\mc{Z}^* \subseteq \mc{Z}$, the boundedness of $\mc{Z}^*$ follows from boundedness of $\mc{Z}$. The closeness of $\mc{Z}^*$ follows from the characterization of Wardrop equilibria in terms of variational inequalities \cite{sandholm2010population}, from continuity of the delay functions, and from closeness of $\mc{Z}$.
%\end{proof}

\subsection{Logit dynamics}
\label{sec:logit}
While the description made so far is completely static, we now endow routing games with evolutionary dynamics.
These are continuous-time dynamical systems that describe how users revise their decisions.
In this work we focus on the \emph{logit dynamics}.
The logit dynamics arises from the mean-field limit (in the spirit of Kurtz's theorem \cite{Kurtz1981}) of the noisy best response dynamics of classical game theory, which describes users that aim at choosing optimal routes, but sometimes select suboptimal ones due to the presence of noise. Formally, the logit dynamics reads, for every $p \in \mc{P}$ and $i \in \mc{R}_p$,
\begin{equation}
	\dot{z}_i^p = v_p\frac{\exp(- c_i^p(z)/\eta))}{\sum_{j \in \mathcal{R}_p} \exp(-c_j^p(z)/\eta)} - z_i^p,
	\label{eq:logit}
\end{equation}
where $\eta \in(0,+\infty)$ is the \emph{noise level}. We refer to \emph{logit($\eta$)} to denote the continuous-time dynamical system \eqref{eq:logit} for a given value of $\eta$. The value of $\eta$ describes how suboptimal the choices of the users are.
As $\eta \to +\infty$, the users select routes with uniform probability distribution, independently of the route cost, i.e.,
\begin{equation*}
	\dot{z}^p_i=\frac{v_p}{|\mc{R}_p|}-z^p_i.
\end{equation*}
As the noise $\eta$ decreases, the users tend to assign a larger probability to routes with smaller cost. 
%Let $\mc{R}_p^{*}(z)$ denote the set of optimal routes $i \in \mc{R}_p$ under route flow $z$, i.e., the routes such that $c_i^p(z) \le c_j^p(z)$ for every route $j \in \mc{R}_p$. 
In the limit of vanishing $\eta$, the dynamics converge to the best response dynamics,
%\begin{equation*}
%	\dot{z}_i^p = 
%	\begin{cases}
%		\frac{v_p}{|\mc{R}_p^{*}(z)|} - z_i^p, \quad & \text{if} \ i \in \mc{R}_p^{*}(z),\\
%		- z_i^p & \text{otherwise},
%	\end{cases}
%\end{equation*}
where users sample uniformly random among the optimal routes and choose suboptimal ones with zero probability.

It is known that, for homogeneous routing games, i.e., when $|\mc P|=1$, for every $\eta > 0$ the logit dynamics admit a globally asymptotically stable fixed point $z_{\eta}$ and that such $z_{\eta}$ converges to the set of Wardrop equilibria as $\eta$ tends to vanish, i.e., 
$$\lim_{\eta\to 0^+}\dist(z_{\eta},\mc Z^*)=0\,.$$
In contrast, the following example illustrates how much more complex behaviors can emerge in case of heterogeneous congestion games, i.e., when $|\mc P|\ge2$. 

%In this work we characterize the logit dynamics in heterogeneous routing games. To motivate our analysis, in the next section we shall discuss the behaviour of the dynamics in homogeneous routing games.
\begin{example}
\label{sec:example}
\begin{figure}
	\centering
	\begin{tikzpicture}[scale=1, transform shape]
		\node[draw, circle] (1) at (0,0)  {o};
		\node[draw, circle] (2) at (2,1)  {a};
		\node[draw, circle] (3) at (2,-1)  {b};
		\node[draw, circle] (4) at (4,0) {d};
		
		\path (1) edge [->] node [above] {$e_1$} (2);
		\path (2) edge [->, bend left = 20] node [above] {$e_2$} (4);
		\path (2) edge [->, bend right = 20] node [above] {$e_3$} (4);
		\path (1) edge [->] node [above] {$e_4$} (3);
		\path (3) edge [->, bend right = 20] node [above] {$e_6$} (4);
		\path (3) edge [->, bend left = 20] node [above] {$e_5$} (4);
	\end{tikzpicture}\\[5pt]
\resizebox{\columnwidth}{!}{
	\begin{tabular}{c|c|c|c|c|c|c|c}
		$p$ & $v_p$ & $\tau_1^p(f_1)$ & $\tau_2^p(f_2)$ & $\tau_3^p(f_3)$ & $\tau_4^p(f_4)$ & $\tau_5^p(f_5)$ & $\tau_6^p(f_6)$ \\
		\hline
		1 & 1.2 & 19+$f_1$ & 19+$f_2$ & 100 & 19+$f_4$ & 100 & 19+$f_6$ \\ 
		2 & 1 & 19+$f_1$ & 20$f_2$ & 100 & 19+$f_4$ & 21+$f_5$ & 100 \\ 
		3 & 1 & 19+$f_1$ & 100 & 21+$f_3$ & 19+$f_4$ & 100 & 20$f_6$ \\ 
	\end{tabular}}
	\caption[A heterogeneous routing games possessing multiple Wardrop equilibria]{A heterogeneous routing games possessing multiple Wardrop equilibria \cite{konishi2004uniqueness}. %A blank in the table means that the delay associated to the link is so high that is never rational to use it \cite{konishi2004uniqueness}.
	}
	\label{fig:no_uniqueness}
\end{figure}
\begin{figure}
	\centering
	\includegraphics[width=7.5cm]{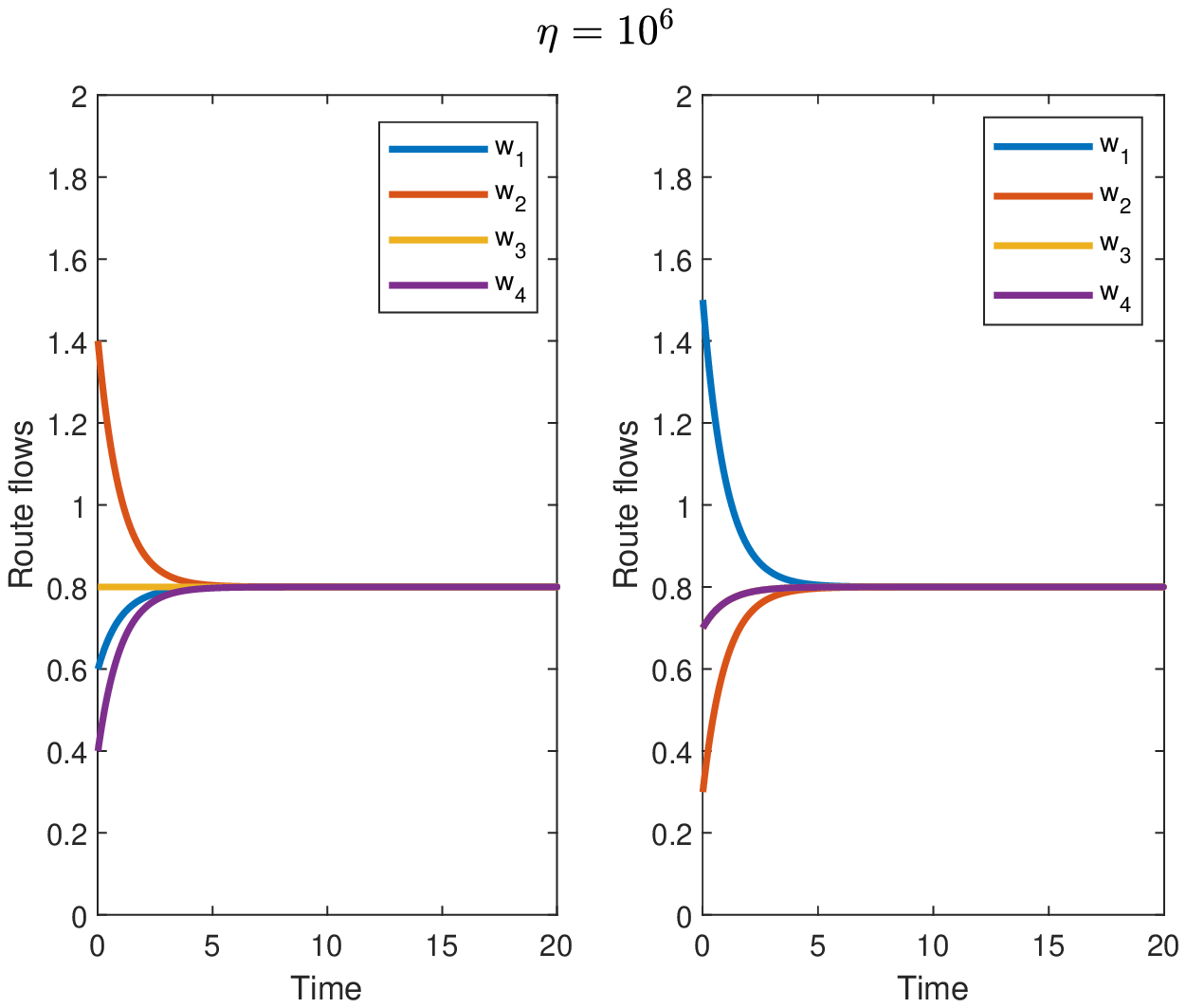}
	\includegraphics[width=7.5cm]{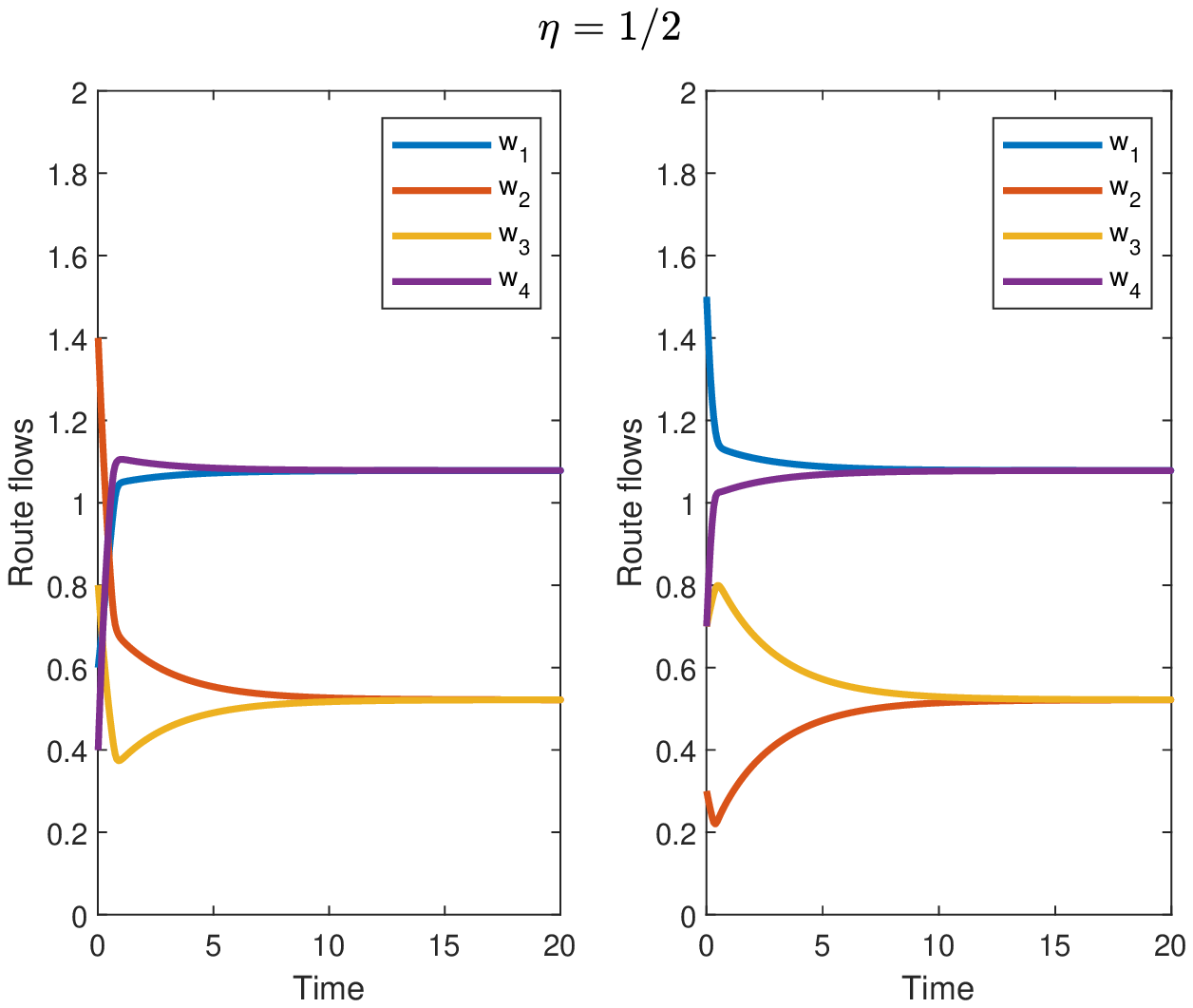}\\
	\includegraphics[width=7.5cm]{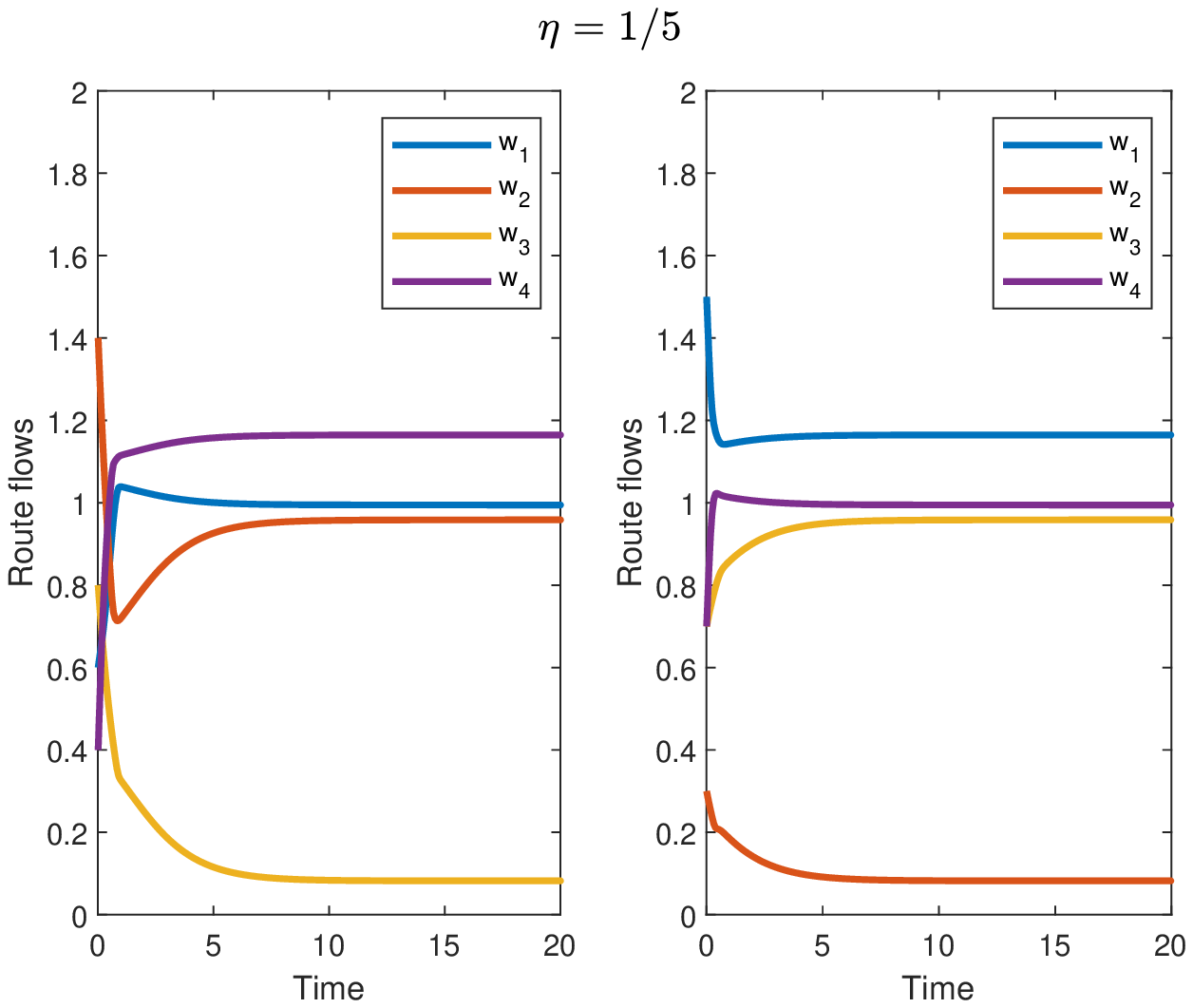}
	\includegraphics[width=7.5cm]{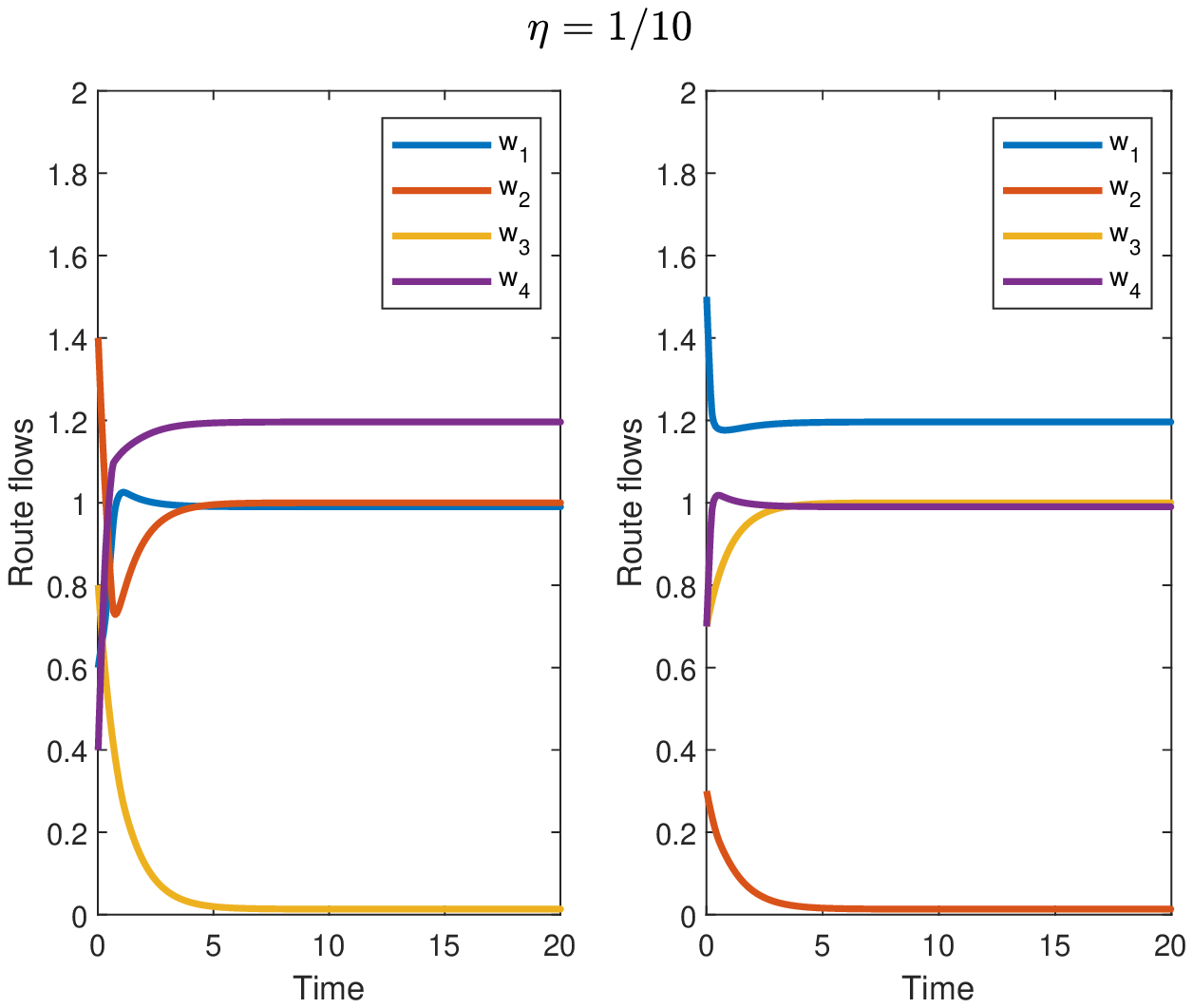}
	\caption[Numerical simulations of the logit dynamics for several noise levels]{Numerical simulations of the logit dynamics for heterogeneous routing game in Example \ref{sec:example}. For every value of $\eta$, we plot two trajectories corresponding to different initial conditions. The trajectories are projected in the space of the aggregate route flow $w$.} \label{fig:sim_eta}
\end{figure}
\begin{figure}
	\centering
	\includegraphics[width=8cm]{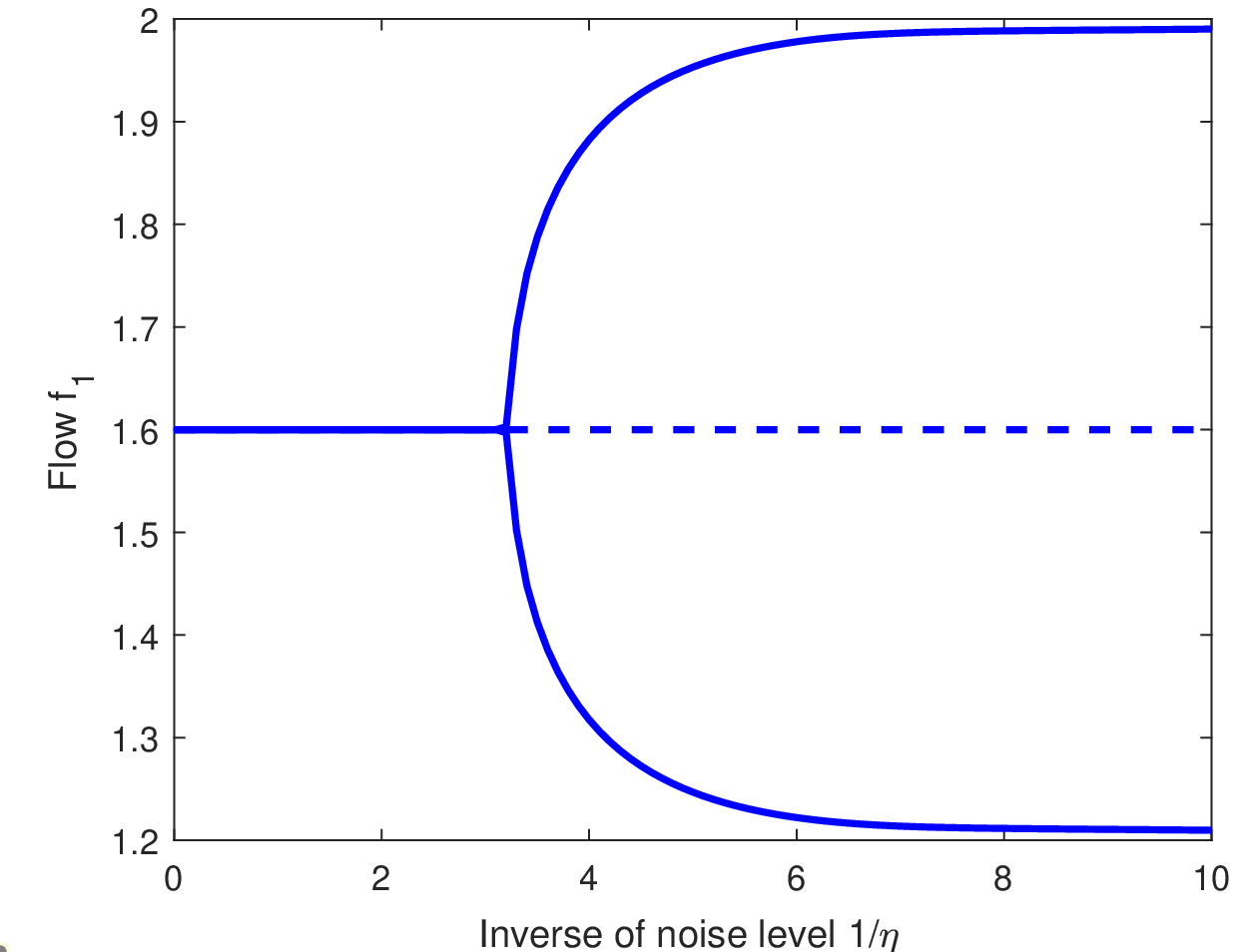}
	\caption[Bifurcation of the logit dynamics]{Bifurcation diagram of logit$(\eta)$ of Example \ref{sec:example}. For simplicity we plot the first component of link flow only, but similar diagrams may be shown for the other components. \label{fig:bifurcation}}
\end{figure}

Consider the heterogeneous routing game in Figure \ref{fig:no_uniqueness} (due to \cite{konishi2004uniqueness}).  We assume that all the populations have the same origin-destination pair $(\oo,\dd)$, and let $$r_1 = (e_1,e_2), \ r_2=(e_1,e_3), \ r_3=(e_4,e_5),\  r_4=(e_4,e_6)$$ the routes from $\oo$ to $\dd$. %Note from the delay functions in Figure \ref{fig:no_uniqueness} that each population has two available routes. In particular, population $1$ can use routes $r_1$ and $r_4$, population $2$ can use routes $r_1$ and $r_3$, and population $3$ can use routes $r_2$ and $r_4$. 
By some computations, one can prove the existence of the following Wardrop equilibria:\\
\begin{enumerate}
	\item
	$
	\begin{cases}
		z_1^1 = 1.2, \ z_4^1 = 0\\
		z_1^2 = 0, \ z_3^2= 1\\
		z_2^3 = 0, \ z_4^3= 1
	\end{cases}
	$\\[5pt]
	\item$
	\begin{cases}
		z_1^1 = 0, \ z_4^1 = 1.2\\
		z_1^2 = 1, \ z^2_3 = 0\\
		z_2^3 = 1, z_4^3 = 0
	\end{cases}
	$\\[5pt]
	\item
	$
	\begin{cases}
		z_{1}^1 = 3/5, \ z_{4}^1=3/5\\
		z^2_{1}=10/21, \ z_{3}^2=11/21\\
		z^3_{2} =11/21, \ z^3_{4}=10/21.
	\end{cases}
	$\\[5pt]
\end{enumerate}

By plugging the equilibria flows in the cost functions one can show that the first two equilibria are strict.
Figure \ref{fig:sim_eta} provides numerical simulations of the logit dynamics for this example. The simulations are conducted with four different values of $\eta$, and two trajectories corresponding to different initial conditions are illustrated, projected onto the space of the aggregate route flows $w = \sum_p z^p$ (notice that $w$ is well defined in this example, since all the populations have same origin-destination pair and route set). As $\eta = 10^6$ (large noise limit) both the trajectories converge to a fixed point in which all the populations distribute uniformly over the route set.
As the noise decreases $(\eta = 1/2)$, the asymptotic state of the system varies, but the trajectories still converge to a unique fixed point. For smaller $\eta$, the system exhibits a bifurcation. Specifically, the two trajectories converge to different fixed points, which approach the two strict equilibria of the game as $\eta$ decreases. We observe from Figure \ref{fig:bifurcation} that the system exhibits a pitchfork bifurcation. By numerical simulations one can observe that the critical value for the bifurcation is $\eta^* \simeq 0.31$. If $\eta > \eta^*$, the system admits a globally asymptotically stable fixed point. If $\eta < \eta^*$, such a fixed point becomes unstable, and two stable fixed points approaching the strict equilibria arise. The unstable fixed point converges to the third Wardrop equilibrium as $\eta$ tends to vanish, thus showing that all the Wardrop equilibria are accumulation points of sequence of fixed points of the dynamics, despite the third one being unstable. In the next section we shall provide our theoretical results, which formalize some of the observations of this section.
\end{example}

\section{Main results}
\label{sec:logit_routing}
%In this section we analyse the behaviour of the logit dynamics in heterogeneous routing games. %We first provide in Section \ref{sec:example} a motivating example to illustrate the richness of behaviours that the logit dynamics may exhibit in heterogeneous routing games. Then, in Section \ref{sec:evo_no_pot} we establish our novel results on the logit dynamics in heterogeneous routing games, which characterize the dynamics both in the large and vanishing noise regimes.
%\subsection{Logit in heterogeneous games: a motivating example}
%\subsection{Logit in heterogeneous games: results}
%\label{sec:evo_no_pot}
%Motivated by the numerical example of Section \ref{sec:example}, i
In this section we present our main results characterizing the fixed points of the logit dynamics and their stability in heterogeneous routing games. 

Our first result shows that the set of the fixed points is non-empty and compact for every noise level, and that the fixed points approach a subset of Wardrop equilibria of the game (called \emph{limit equilibria}) in the limit of vanishing noise. Moreover, we show that every strict equilibrium belongs to the set of the limit equilibria.
%\textcolor{blue}{Furthermore, we show that Wardrop equilibria satisfying certain conditions are always approximated by fixed point of the logit dynamics. These conditions do not depend on the network topology, but on properties of the examined Wardrop equilibrium. To this end, we introduce the following definitions.}
%\begin{definition}[Quasistrict equilibrium]
%	An equilibrium $\mb{z}$ is called \emph{quasistrict} for population $p$ if for every route $r,s$ such that $z_r^p>z_s^p=0$, it holds $c^p_r(\mb{z}) < c^p_s(\mb{z})$. The equilibrium is called \emph{quasistrict} if it is quasistrict for every population $p$.
%\end{definition}
%Let us start with the plan. The next proposition states that in the limit of infinite noise every fixed point of the logit dynamics approaches the set of the Wardrop equilibria of the game.
In order to formulate the results properly, for every $\eta \in (0,+\infty)$ we let $\Omega_\eta \subseteq \mc{Z}$ denote the set of fixed points of logit$(\eta)$, and let $\Omega_0$ denote the set of accumulation points of convergent sequences of fixed points of the logit dynamics as the noise vanishes, i.e.,
$$\Omega_0=\left\{z\in\mc Z: \exists \ (\eta_n)_n, \eta_n \to 0, z_{n} \in\Omega_{\eta_n}, z_{n}\to z\right\}.$$
%we introduce the following (standard) notion of $\omega$-limit set. 

%\begin{definition}
%Let $z_0 \in \mc{Z}$, and $F_\eta(t,z_0)$ be the solution of logit$(\eta)$ at time $t$ corresponding to initial condition $z(0)=z_0$. The \emph{$\omega_\eta$-limit set} of $z_0$, denoted by $\omega_\eta(z_0)$ is the set containing all $z \in \mc{Z}$ such that, for an unbounded sequence $t_k$,
%$$
%z = F_\eta(t_k,z_0).
%$$
%We denote by $\omega_\eta$ the union of all $\omega_\eta(z_0)$ for all $z_0 \in \mc{Z}$.
%\end{definition}

\begin{theorem}
	\label{thm:logit}
	Let $\mc{Z}^* \subseteq \mc{Z}$ be the set of Wardrop equilibria of a heterogeneous routing game, and consider the associated dynamics logit$(\eta)$ defined in \eqref{eq:logit}. Then:
	%\begin{equation}
	%	\dot{z}_i^p = \tau^p\frac{\exp(-\eta\cdot c_i^p(\mb{z}))}{\sum_{j \in \mathcal{R}} \exp(-\eta\cdot c_j^p(\mb{z}))} - z_i^p, \quad \forall p \in \mc{P}, i \in \mc{R}.
	%	\label{eq:logit2}
	%\end{equation}
	\begin{enumerate}
		\item[(i)] $\Omega_\eta$ is non-empty and compact for every $\eta > 0$;
		%\item[(ii)] there exists a non-empty compact set $\Omega_0 \subseteq \mc{Z}^*$ such that 
		%$$\lim_{\eta \to 0^+} \omega_\eta = \Omega_0;$$
		%\item[(ii)] let $(z_{n})_n$ and $(\eta_{n})_n$ two convergent sequences such that $\lim_{n \to +\infty} \eta_{n} = 0$, and $z_{n} \in \Omega_{\eta_n}$. Then, there exists a non-empty %compact
		%set $\Omega_0 \subseteq \mc{Z}^*$ such that
		%$$\lim_{n \to +\infty} z_{n} \in \Omega_0;$$
		\item[(ii)] $\Omega_0$ is a non-empty compact subset of the Wardrop equilibria, i.e., $$\Omega_0 \subseteq \mc{Z}^*;$$
		%		where the convergence of compact sets is meant in the sense of \footnote{Let $\mc{X}$ be a compact space and $(\mc{X}_n)_n$ be a sequence of compact sets. We say that $\lim_{n \to +\infty} \mc{X}_n = \mc{X}$ if the following two conditions hold: \emph{i)} for every $x \in \mc{X}$, there exists a sequence $(x_n)_n$ such that $x_n \in \mc{X}_n$ for every $n$ and $\lim_{n \to +\infty} x_n = x$; \emph{ii)} for every converging sequence $(x_n)_n$, with $x_n \in \mc{X}_n$ for every $n$, $\lim_{n \to +\infty} x_n = x \in \mc{X}$.}.
		%Definition \ref{def:set_lim} (\textcolor{red}{dove definisco la convergenza tra insiemi?}).
		\item[(iii)] all strict Wardrop equilibria (if any) belong to $\Omega_0$, i.e., 
		$$\mc Z _s^* \subseteq \Omega_0.$$
		Moreover, for every strict equilibrium $z^* \in \mc{Z}_s^*$, there exists $\tilde{\eta}>0$ and a family of vectors $(z_{\eta})_{\eta < \tilde\eta}\subseteq\mc Z$ such that 
		$$
		\lim_{\eta\to 0^+}z_{\eta}=z^*,
		$$ 
		with $z_{\eta}$ asymptotically stable fixed point of logit$(\eta)$. 
	\end{enumerate}
	%For every value of $\eta \ge 0$, let $\Omega_\eta \subseteq \mc{Z}$ denote the set of fixed points of the dynamics as function of $\eta$. $\Omega_\eta$ is non-empty and compact. Moreover, consider any sequence $\eta_n \subseteq[0,+\infty)$ such that $\lim_{n\to+\infty}\eta_n=+\infty$. Then, $\lim_{n \to +\infty} \Omega_{\eta_n} = \overline{\mc{Z}}^*$, with $\overline{\mc{Z}}^* \subseteq \mc{Z}^*$. Furthermore, every strict Wardrop equilibrium belongs to $\mc{\overline{Z}}^*$.
	%in the limit of vanishing noise, $\lim_{\eta +\rightarrow \infty} \Omega_\eta = \overline{\mc{Z}}^*$, with $\overline{\mc{Z}}^* \subseteq \mc{Z}^*$.
\end{theorem}
\begin{proof}
	See Appendix \ref{app1}.
\end{proof}

Throughout the paper, we shall refer to $\Omega_0$ as the set of \emph{limit equilibria} of the routing game. Theorem \ref{thm:logit} states that the set of limit equilibria is a nonempty compact subset of Wardrop equilibria that includes all strict Wardrop equilibria (if any). 
%, but not all the Wardrop equilibria of the game are in $\overline{\mc{Z}^*}$. A complete characterization of $\overline{\mc{Z}^*}$ is still missing, but Theorem \ref{thm:logit} states that all the strict equilibria are in $\overline{\mc{Z}^*}$, generalizing the observations of Remark \ref{remark:strict} to the case of heterogeneous routing games.
Moreover, in addition to being approximated by fixed points of the logit dynamics, strict equilibria are also locally asymptotically stable under the dynamics in the vanishing noise limit.
%states that in the limit of vanishing noise the strict equilibria are also locally asymptotically stable under the logit dynamics, where stability of Wardrop equilibria has to be meant as stability of fixed points approaching it.
\begin{remark}
	This result must be compared with the existing literature. It is known that interior evolutionary stable states $z$ of populations game admit a neighborhood of $z$ such that, for large enough $\eta$, there exists one and only one fixed point of logit$(\eta)$ \cite[Theorem 8.4.6]{sandholm2010population}. Moreover, such fixed points are locally asymptotically stable in the limit of vanishing noise. Although strict equilibria are evolutionary stable states, they are not interior, thus violating one of the assumptions of \cite[Theorem 8.4.6]{sandholm2010population} and making our result original.
\end{remark}

In the next part of this section
we investigate the asymptotic behaviour of the logit dynamics in the large noise limit.
The next result states that in this regime the logit dynamics admits a globally exponentially stable fixed point for every routing game.
\begin{theorem}
	\label{thm:noise}
	Let ($\mathcal{G}$, $\mathcal{P}$, $\tau,v$) be a heterogeneous routing game, and consider the corresponding logit$(\eta)$ defined in \eqref{eq:logit}. Then, there exists $\underline{\eta} > 0$ such that logit$(\eta)$ admits a globally exponentially stable fixed point for every $\eta \in (\underline\eta,+\infty)$. %for every $k \in (0,1]$, there exist $\eta_k$ such that for $\eta \in [0,\eta_k]$, logit$(\eta)$ admits a globally exponentially stable fixed point with rate $k$.
\end{theorem}
\begin{proof}
	See Appendix \ref{app3}.
\end{proof}
%Theorem \ref{thm:noise} characterizes the behaviour of the logit dynamics in the large noise regime, showing that the dynamics admits a globally asymptotically stable. 

The results established in Theorems \ref{thm:logit}-\ref{thm:noise} characterize the behaviour of the logit dynamics in heterogeneous routing games independently of the network topology, and explain the numerical simulations of Example \ref{sec:example}. In particular, the theorems suggest that, if a heterogeneous routing game admits multiple strict equilibria, then the logit dynamics admit a bifurcation, as shown in Example \ref{sec:example}.
We conclude this section with some remarks.
\begin{remark}
\label{remark:strict}
The behaviour of the logit dynamics in heterogeneous routing games must be compared with stability results in homogeneous routing games. The asymptotic global stability of equilibria in homogeneous routing games relies on the fact that homogeneous routing games admit a convex potential function $V(z)$ \cite{beckmann1956studies}, i.e., Wardrop equilibria $\wz$ correspond to solutions of the convex program
$$
\wz \in \argmin_{z \in \mc{Z}} V(z).
$$
The existence of a convex potential implies that 
$$
V_{\eta}(z):=V(z)+\eta\sum_{i \in \mc{R}} z_i \log\left(\frac{z_i}{v}\right),
$$
is a strictly convex Lyapunov function of logit($\eta$), hence the unique minimizer of $V_\eta$, denoted by $z_\eta$, is globally attractive for the dynamics. As $\eta$ tends to vanish, $\lim_{\eta \to 0^+} \text{dist}(z_\eta,\mc{Z}^*) = 0$, i.e., the asymptotically globally stable fixed of the dynamics converges to the set of the Wardrop equilibria of the game.
Observe that, since the potential $V(z)$ is convex, the set of the Wardrop equilibria is convex.
However, as stated for heterogeneous routing games in Theorem \ref{thm:logit}-(ii), the fixed points of the dynamics approach the set of the Wardrop equilibria of the game as the noise vanishes, but not every Wardrop equilibrium is an accumulation point of fixed points of the logit dynamics. Strict equilibria play a special role also in homogeneous routing games. Indeed, if a homogeneous routing game game admits a strict equilibrium $\wz$, then $\wz$ is isolated and $\mc{Z}^*$ is a singleton. Therefore, if a homogeneous routing game admits a strict equilibrium $\wz$, then $\wz$ is globally asymptotically stable in the vanishing noise limit. % We will derive a similar (but weaker) result on the stability of strict equilibria in heterogeneous routing games in Section \ref{sec:evo_no_pot}.
Observe that the local asymptotic stability of strict equilibria in heterogeneous routing games established in Theorem \ref{thm:logit}-(iii) is a weaker result compared to the global asymptotic stability established for homogeneous routing games. We remark that this limitation is an intrinsic property of heterogeneous routing games. Indeed, as illustrated in Example \ref{sec:example}, heterogeneous routing games may admit multiple strict equilibria, hence global asymptotic stability of strict equilibria does not hold in general.
\end{remark}
\begin{remark}
	Similar considerations as in Remark \ref{remark:strict} apply to heterogeneous routing games that admit a convex potential. While in general heterogeneous routing games do not admit a potential function \cite{sandholm2010population,dafermos1971extended}, if the delay functions satisfy the symmetry condition
	\begin{equation}
		\label{eq:symmetry}
		\sum_{e \in i \cap j} (\tau_e^p)' = \sum_{e \in i \cap j} (\tau_e^q)' \quad \forall  p,q \in \mc{P}, i \in \mc{R}_p, j \in \mc{R}_q,
	\end{equation}
    then the game admits a potential function.
	Such a condition is satisfied for instance if the populations differ only in the origin-destination pair, or if constant tolls are charged and the populations differ in the toll sensitivity, i.e., the delay functions (including tolls) are in the form
	\begin{equation}
		\label{eq:toll_sens}
		\tau_e^p(f_e) = \tau_e(f_e) + \alpha_p \omega_e.
	\end{equation}
	%Among those applications, we mention the case of heterogeneous users populations that differ in the toll sensitivity, i.e., delay functions in the form \eqref{eq:toll_sens}. 
	Indeed, one can prove that
	$$
	V(z) = \sum_{e \in \mc{E}} \int_0^{\sum_{p} (A^p z^{p})_e } \tau_e(s) ds + \sum_{p \in \mc{P}} \sum_{e \in \mc{E}} \alpha_p \omega_e (A^p z^p)_e
	$$
	is a convex potential function for this class of games. While in \cite{cole2003pricing,fleischer2004tolls} the existence and characterization of optimal tolls for this class of heterogeneous games are provided, the existence of a convex potential function guarantees that optimal flows are globally asymptotically stable under the logit dynamics when optimal tolls are charged. The existence of a potential function is lost if tolls are in feedback form instead of constant.
\end{remark}

\section{Conclusions and future research}
\label{sec:conclusion}
In this paper we investigate the asymptotic behaviour of the logit dynamics in heterogeneous routing games. We show that fixed points of the logit dynamics converge to a subset of Wardrop equilibria (called limit equilibria) in the vanishing noise limit, and that the set of the limit equilibria include all the strict equilibria of the game. Additionally, we show that strict equilibria are locally asymptotically stable in the vanishing noise limit. Finally, we show that the dynamics admits a globally asymptotically stable fixed point in the large noise limit. Those results together suggest that if a heterogeneous routing game admits multiple strict equilibria, then the logit dynamics exhibits a bifurcation as the noise varies, as shown in the numerical example of Example \ref{sec:example}.

Future research lines include the complete characterization of the limit equilibria in the vanishing noise limit. Our conjecture is that every connected component of equilibria admits one and only one limit equilibrium of the logit dynamics. While strict equilibria have been proven to be locally asymptotically stable, another open issue is a complete characterization of the asymptotically stable equilibria of heterogeneous routing game. 

Another interesting direction is the application of our theoretical results for the analysis of multi-scale dynamics, the design of dynamic feedback tolls, and the optimization of network interventions in heterogeneous congestion games. 
%, with the goal of minimizing the social cost of the equilibria while ensuring stability of the system. 
While these problems have been addressed, e.g., in \cite{Como.ea:2013}, \cite{Como.Maggistro:22}, and \cite{cianfanelli2021optimal}, respectively, for homogeneous preferences, we are not aware of extensions of these results to the heterogeneous case. 

\section{Acknowledgements}This work was partly supported by the Italian Ministry for University and Research through grants ``Dipartimenti
di Eccellenza 2018–2022'' [CUP: E11G18000350001] and Project PRIN 2017 ``Advanced Network Control of Future Smart Grids''
(http://vectors.dieti.unina.it), and by the {\it Fondazione Compagnia di San Paolo} through a Joint Research Project and  Project ``SMaILE''.

\bibliographystyle{IEEEtran}
\bibliography{references.bib}

\appendices

\section{Proof of Theorem \ref{thm:logit}}
\label{app1}
\emph{(i)} Consider the function $G:\mc{Z} \times (0,+\infty) \to \mc{Z}$, with components
\begin{equation*}
	G_r^p(z,\eta)=v_p\frac{\exp(-c_r^p(z)/\eta)}{\sum_{s \in \mathcal{R}_p} \exp(-c_s^p(z)/\eta)} \quad \forall p \in \mc{P}, r \in \mc{R}_p.
\end{equation*}
Notice that logit($\eta$) reads
\begin{equation}
	\label{eq:logit_short}
	\dot{z}_r^p = G_r^p(z,\eta)-z_r^p,
\end{equation}
hence elements of $\Omega_\eta$ coincide with fixed points of $G(\cdot,\eta)$. Observe also that for every $\eta > 0$, $G(\cdot,\eta)$ is continuous and maps the non-empty compact convex set $\mc{Z}$ in itself. Hence, Brouwer's fixed point theorem guarantees that $G(\cdot, \eta)$ admits at least one fixed point in $\mc{Z}$ \cite{border1985fixed}, i.e., the set of fixed points $\Omega_{\eta}$ is non-empty for every $\eta > 0$. Notice that $\Omega_{\eta}$ is bounded because $\mc{Z}$ is bounded and $\Omega_\eta \subseteq \mc{Z}$. Moreover, since it is the level set of a continuous function, $\Omega_{\eta}$ is closed. Therefore, the set $\Omega_{\eta}$ is compact for every $\eta > 0$.

\emph{(ii)}
First, observe that $\Omega_0$ is non-empty because $\mc Z$ is bounded, thus every sequence of elements in $\mc Z$ admits a converging subsequence.
Consider $\wz \in \Omega_0$, and the corresponding sequences $\eta_n$ (with $\eta_n \stackrel{n \to +\infty}{\longrightarrow} 0$) and $z_{n} \in\Omega_{\eta_n}$ (with $z_{n}\stackrel{n \to +\infty}{\longrightarrow} \wz$).
Consider a suboptimal route $r$ for population $p$ under $\wz$, i.e., a route $r \in \mc{R}_p$ such that $c_s^p(\wz)<c_r^p(\wz)$ for some $s \in \mc{R}_p$. Then, 
\begin{equation}
	\lim_{n \rightarrow +\infty} G_r^p(z_n,\eta_n) = 0.
	\label{eq:lim}
\end{equation}
%where with a slight abuse of notation we replace $\lim_{k \to +\infty} \eta_{n_k}$ with $\lim_{\eta \to +\infty} \eta$.
From \eqref{eq:logit_short} and \eqref{eq:lim} it follows
\begin{equation*}
	\wz_r^p=\lim_{n \to +\infty}(z_{n})_r^p
	= \lim_{n \to +\infty} G_r^p(z_n,\eta_n)=  
	0.
\end{equation*}
Hence, $\wz$ is a Wardrop equilibrium according to Definition \ref{def:wardrop}. This implies in particular that $\Omega_0\subseteq\mc Z^*$.
Since $\Omega_0$ is bounded, to establish compactness of $\Omega_0$ we need to prove that $\Omega_0$ is closed. To this end, consider a converging sequence $(z_k)_k$ with $z_k \in \Omega_0$ for every $k$, and let $\wz$ denote the limit of the sequence, i.e., $z_k \stackrel{k \to +\infty}{\longrightarrow} \wz$. Our goal is to prove that $\wz \in \Omega_0$. By definition of $\Omega_0$, for every $k$, there exist two convergence sequences $(\eta_m)_m$ (with $\eta_m \stackrel{m \to +\infty}{\longrightarrow} 0$) and $(z_k^m)_m$ (with $z_k^m \in \Omega_{\eta_m})$ such that $z_k^m \stackrel{m \to +\infty}{\longrightarrow} z_k$.
This implies the existence of $m_k$ such that
$$
|z_k^{m_k}-z_k| < \frac{1}{k}.
$$
We now prove that $z_k^{m_k} \stackrel{k \to +\infty}{\longrightarrow} \wz$. For every $\epsilon>0$, we take $k$ such that $|z_k - \wz| < \epsilon/2$ and $k>2/\epsilon$. Then,
\begin{equation}
	\label{eq:omega0}
|z_k^{m_k} - \wz| \le |z_k^{m_k} - z_k| + |z_k - \wz| < \frac{1}{k} + \frac{\epsilon}{2} < \epsilon.
\end{equation}
Since $z_k^{m_k}$ is by construction a fixed point of logit$(\eta_{m_k})$, \eqref{eq:omega0} shows that $\wz$ is an accumulation point of fixed points of the logit dynamics and thus is contained in $\Omega_0$, concluding the proof.

\emph{(iii)} Consider a strict equilibrium $\wz \in \mc{Z}^*_s$, and let $r_p$ denote the optimal route for population $p$, i.e., $\wz^p = v_p \delta^{(r_p)}$ for every $p \in \mc{P}$. For every $\epsilon\ge0$, let
\begin{equation*}
	O_\epsilon = \{z \in \mc{Z}:  z^{p}_{r_p} \ge v_p(1-\epsilon) \ \forall p \in \mc{P}\},
\end{equation*} 
be the set of route flows such that at least a fraction $1-\epsilon$ of agents of every population $p$ use its optimal route $r_p$. Note that $\wz \in O_{\epsilon}$ for every $\epsilon \ge 0$. Let
\begin{equation*}
	\alpha := \min_{p \in \mc{P}}\min_{s \in \mc{R}_p \setminus \{r_p\}}[c_s^p(\wz)-c_{r_p}^p(\wz)]>0.
\end{equation*}
Note that $\alpha>0$ is a consequence of $\wz$ being strict.
We then define $\overline{\epsilon}$ to be the largest $\epsilon$ such that for every $z \in O_\epsilon$, for every population $p$ and route $s \in \mc{R}_p \setminus \{r_p\}$, the difference between the cost of route $s$ and the cost of route $r_p$ is at least $\alpha/2$, i.e.,
\begin{equation*}
	\overline{\epsilon}=\max\Big\{\epsilon \ge 0: \min_{z \in O_\epsilon} \ \min_{p \in \mc{P}}\min_{s \in \mc{R}_p \setminus \{r_p\}} [c_s^p(z)-c_{r_p}^p(z)]\ge \frac{\alpha}{2}\Big\}.
\end{equation*} 
Note that $\overline{\epsilon}>0$, since the equilibrium is strict. We now show that for every $\epsilon \in (0,\overline{\epsilon}]$ there exists $\eta_\epsilon$ such that, for every $\eta\in (0,\eta_\epsilon]$, $G(\cdot,\eta)$ maps $O_\epsilon$ in itself. To this end, observe that for every $\epsilon \in (0, \overline{\epsilon}]$ and population $p \in \mc P$, the route $r_p$ is strictly optimal for every flow $z \in O_\epsilon$ by construction (recall the definition of $\overline{\epsilon}$), i.e.,
\begin{equation}
	\label{eq:strict_intorno}
	c_{r_p}(z) < c_s(z) \quad \forall z \in O_\epsilon, p \in \mc{P}, s \in \mc{R}_p \setminus \{r_p\}. 
\end{equation}
Thus, for every $z \in O_\epsilon$ and population $p \in \mc{P}$
\begin{equation}
	\label{eq:lim_eta}
	\lim_{\eta \to 0^+}G_{i}^p(z,\eta) = \begin{cases}
		v_p \quad &\text{if} \ i=r_p,\\
		0 \quad &\text{if} \ i \in \mc{R}_p \setminus \{r_p\}.
	\end{cases}
\end{equation}
Note that the right term of \eqref{eq:lim_eta} corresponds to $\wz$.
since $O_\epsilon$ tends to $\wz$ as $\epsilon \to 0$, this implies by continuity of $G$ in $\eta$ that for every $\epsilon \in (0,\overline{\epsilon}]$ there exists a small enough $\eta_\epsilon$ such that for $\eta \in (0,\eta_\epsilon]$, $G(\cdot,\eta)$ maps $O_\epsilon$ in itself. Since $O_\epsilon$ is compact and convex, Brower's fixed point theorem ensures the existence of at least a fixed point of $G(\cdot,\eta)$ in $O_\epsilon$ for every $\eta \in (0,\eta_\epsilon]$. Since the argument holds for every small enough $\epsilon$, then there exists a sequence of fixed points $z_n$ such that $\lim_{n \to + \infty} z_n = \wz$, showing that $\wz \in \Omega_0$.
To prove the second part of (iii), %let us consider a strict equilibrium $\tilde{z} \in \mc{Z}_s^*$, and let $r_p$ denote the optimal route for population $p$, i.e., $\wz^p = v_p \delta^{(r_p)}$ for every $p \in \mc P$. We define $O_\epsilon$, $\alpha$ and $\overline{\epsilon} > 0$ following the same steps as before. By definition of $\overline{\epsilon}$, we have that for every $\epsilon \in (0, \overline{\epsilon}]$ and  $z \in O_{\epsilon}$, the route $r_p$ is strictly optimal for population $p \in \mc{P}$, i.e.,
let us write the logit dynamics in the form 
$$
\dot{z} = G(z,\eta)-z,
$$
with $G: \mc{Z} \times (0,+\infty) \to \mc{Z}$. %has components
%\begin{equation*}
%	G_i^p(z,\eta) = v_p\frac{\exp(-c_i^p(z)/\eta)}{\sum_{j \in \mathcal{R}_p} \exp(-c_j^p(z)/\eta)}.
%\end{equation*}
We now extend $G$ to include the limit value $\eta=0$, while restricting $z \in O_{\overline\epsilon}$. Formally, let us define $\overline{G} : O_{\overline{\epsilon}} \times [0,+\infty) \to \mc{Z}$ as
$$
\overline{G}(z,\eta) = 
\begin{cases}
	G(z,\eta) \quad & \text{if} \ \eta>0 \\
	\wz & \text{if} \  \eta=0.	
\end{cases}
$$
We now show that $\overline{G} \in \mc C^1$. %First, observe that imposing $z \in O_{\overline\epsilon}$ ensures that the optimal route for every population $p$ is $r_p$ for every $z \in O_{\overline{\epsilon}}$. Indeed, the definition of $\overline{\epsilon}$ ensures that for every the obtained function is continuously differentiable. 
The continuity follows from \eqref{eq:strict_intorno}, which implies that $\lim_{\eta \to 0^+} G(z,\eta) = \wz$ for every $z \in O_{\overline\epsilon}$ (see \eqref{eq:lim_eta} for details). Restricting the route flow space from $\mc Z$ to $O_{\overline\epsilon}$ is needed to ensure that the optimal route is unique and independent of $z$ for all populations, thus avoiding discontinuities of $\overline G(z,0)$ in $z$. To prove continuous differentiability, let us define $\Delta_{si}^p(z) = c_s^p(z)-c_i^p(z)$, and let $J_{G,z}(z,\eta)$ denote the Jacobian of $G(z,\eta)$ with respect to $z$. Since $G \in \mc C^1$, to prove that $\overline G \in \mc C^1$ we must investigate its differentiability in $\eta=0$. We first show that $\overline{G}(z,\eta)$ is continuously differentiable with respect to $z$. To this end, our goal is to prove that
\begin{equation}
	\label{eq:J}
	\lim_{\eta \to 0^+}J_{G,z}(z,\eta) = \mb{0}, \quad \forall z \in O_{\overline\epsilon},
\end{equation} 
where the components of $J_{G,z}(z,\eta)$ read
\begin{equation}
	\begin{aligned}
		\label{eq:jac3}
		\frac{\partial G^p_i (z,\eta)}{\partial z^q_j}=\,&
		\frac{v_p}{\eta}\frac{e^{-c_i^p(z)/\eta}\sum_{s \in \mc R_p \setminus \{i\}} \frac{\partial \Delta_{si}^p(z)}{\partial z_j^{q}} e^{- c_s^p(z)/\eta}}{(\sum_{r \in \mc R_p} e^{-c_r^p(z)/\eta})^2}.
	\end{aligned}
\end{equation}
Observe that for every $p,q \in \mc{P}$, $i \in \mc{R}_p \setminus \{r_p\}$, $s \in \mc{R}_p$ and $z \in O_{\overline\epsilon}$ it follows from \eqref{eq:strict_intorno} that
$$
\lim_{\eta \to 0^+} \frac{e^{-c_i^p(z)/\eta}e^{-c_s^p(z)/\eta}}{(\sum_{r \in \mc R_p} e^{-c_r^p(z)/\eta})^2} = 0.
$$
The previous equation implies by \eqref{eq:jac3} that for every $z \in O_{\overline\epsilon}$, if $i \in \mc{R}_p \setminus \{r_p\}$, then
\begin{equation}
	\label{eq:piece1}
	\lim_{\eta \to 0^+}\frac{\partial G^p_i (z,\eta)}{\partial z^q_j} = 0, \quad \forall p,q \in \mc{P}, i \in \mc{R}_p \setminus \{r_p\}, j \in \mc{R}_q.
\end{equation}
On the other hand, since for every $p,q \in \mc{P}$, $s \in \mc{R}_p \setminus \{r_p\}$ and $z \in O_{\overline\epsilon}$ it holds
$$
\lim_{\eta \to 0^+} \frac{e^{-c_{r_p}(z)/\eta}e^{-c_s^p(z)/\eta}}{(\sum_{r \in \mc R_p} e^{-c_r^p(z)/\eta})^2} = 0,
$$
it follows from \eqref{eq:jac3} that
\begin{equation}
	\label{eq:piece2}
\lim_{\eta \to 0^+} \frac{\partial G^p_{r_p}  (z,\eta)}{\partial z^q_j} = 0, \quad \forall p,q \in \mc{P}, j \in \mc{R}_q, z \in O_{\overline\epsilon}.
\end{equation}
Eq. \eqref{eq:piece1} and \eqref{eq:piece2} prove \eqref{eq:J}, which in turn implies that $\overline{G}(z,\eta)$ is continuously differentiable with respect to $z$ for every $z \in O_{\overline \epsilon}$ and $\eta \in [0,+\infty)$. 
We now prove the continuous differentiability of $\overline G(z,\eta)$ with respect to $\eta$ for every $z \in O_{\overline{\epsilon}}$ and $\eta \in [0,+\infty)$, i.e., we prove that
$$
\lim_{\eta \to 0^+}\frac{\partial G_i^p(z,\eta)}{\partial \eta} = \mb 0, \quad \forall p \in \mc P, i \in \mc R_p, z \in O_{\overline\epsilon}.
$$ Since $G \in \mc C^1$ we just need to discuss the continuous differentiability in $\eta=0$. To this end, let us write the components of the Jacobian in the form
\begin{equation}
	\label{eq:der_eta}
	\frac{\partial G_i^p(z,\eta)}{\partial \eta} = -\frac{v_p \sum_{s \neq i} \Delta_{si}^p(z) e^{-\Delta_{si}^p(z)/\eta}}{\eta^2(1+\sum_{r \neq i} e^{-\Delta_{ri}^p(z)/\eta})^2}.
\end{equation}
Again, we split the analysis in two parts. If $i = r_p$, then we have $\Delta_{si}^p(z)>0$ for every $s \neq i$ and $z \in O_{\overline \epsilon}$, which implies by \eqref{eq:der_eta} that
\begin{equation}
	\label{eq:der_eta_rp}
	\lim_{\eta \to 0^+}\frac{\partial G_{r_p}^p(z,\eta)}{\partial \eta} = 0, \quad \forall p \in \mc{P}, z \in O_{\overline\epsilon}.
\end{equation}
Instead, for every $i \neq r_p$ we get that, as $\eta \to 0^+$, the numerator and the denominator in \eqref{eq:der_eta} are dominated respectively by term with $s = r_p$ and $r = r_p$, with $\Delta_{r_p i}(z) < 0$ for every $z \in O_{\overline\epsilon}$, yielding
\begin{equation}
	\label{eq:der_eta_i}
	\lim_{\eta \to 0^+}\frac{\partial G_{i}^p(z,\eta)}{\partial \eta} = 0 \quad \forall i \in \mc{R}_p \setminus \{r_p\}, z \in O_{\overline\epsilon}.
\end{equation}
Eq. \eqref{eq:der_eta_rp} and \eqref{eq:der_eta_i} imply that $\overline{G}(z,\eta)$ is continuously differentiable with respect to $\eta$, proving that $\overline{G} \in \mc C^1$.
To conclude the proof, let us define $g(z,\eta) = G(z,\eta)-z$ and $\overline{g}(z,\eta) = \overline{G}(z,\eta)-z$. Notice that zeros of $g(\cdot,\eta)$ coincide with elements in $\Omega_\eta$ and $\overline{g}(\wz,0)=0$. The existence of a family of fixed points $(z_{\eta})_{\eta < \tilde\eta}$ such that 
$$
\lim_{\eta\to 0^+}z_{\eta}=\wz,
$$  
follows from the implicit function theorem applied to the function $\overline{g}$ in $(\wz,0)$ \cite{rudin1976principles}. To prove the asymptotic stability of this family of fixed points, notice that
$$
J_{g,z}(z,\eta) = J_{G,z}(z,\eta) - \mb{I}.
$$
Since for every $\epsilon \in (0,\overline{\epsilon}]$ there exists $\eta_\epsilon$ such that $z_\eta \in O_\epsilon$ for every $(0,\eta_\epsilon]$, it follows from \eqref{eq:J} that
$$
\lim_{\eta \to 0^+} J_{g,z}(z_\eta,\eta) = -\mb{I},
$$
which implies linear stability (and then local asymptotic stability)) of $z_\eta$ as $\eta \to 0^+$.

\section{Proof of Theorem \ref{thm:noise}}
\label{app3}
We first establish a result on contractive systems. The result is not original and may be found in \cite{jafarpour2020weak}. Still, we provide an alternative and more intuitive proof. Our proof borrows techniques from \cite[Lemma 5]{lovisari2014stability}, where the authors prove that every monotone diagonally dominant system is $l_1$-contractive. Proposition \ref{prp:contractivity} generalizes this result, proving that the Jacobian with negative diagonally dominant columns is a sufficient condition for $l_1$-contractivity.
\begin{proposition}
	\label{prp:contractivity}
	Let $\dot x=g(x)$ be a continuous-time dynamical system. Assume that $g:\mathds{R}^n \to \mathds{R}^n$ is continuously differentiable in $\mc{X} \subseteq \mathds{R}^n$. Let $J(x)$ denote the Jacobian of $g$, and let 
	$$
	\max_{j \in \{1,...,n\}} \left(J_{jj}(x)+\sum_{i:i\neq j} |J_{ij}(x)|\right) \le -c \quad \forall x \in \mc{X}.
	$$ 
	Assume that $\mc{X}$ is $g$-invariant, and let $x(t)$ and $y(t)$ denote the trajectories at time $t$ corresponding to initial conditions $x(0)=x_0 \in \mc{X}$ and $y(0)=y_0 \in \mc{X}$, respectively. Then,
	\begin{enumerate}
		\item for every $t \ge 0$
		\begin{equation}
			\label{eq:l1_contract}
			||x(t) - y(t)||_1 \le e^{-ct} ||x_0-y_0||_1,
		\end{equation}
		\item There exists a globally exponentially stable fixed point in $\mc{X}$.
		%\item $\mb{z} \to ||\mb{z}-\mb{z}^*||_1$ is a global Lyapunov function.
	\end{enumerate}
\end{proposition} 
\begin{proof}
	For simplicity of notation we omit the dependence on $t$.	By definition of the $l_1$-norm and the linearity of the derivative, we get 
	\begin{equation} \label{eq1}
		\begin{aligned}
			\frac{d}{dt}\|x-y\|_1   = &\frac{d}{dt} \sum_{i} |x_i-y_i|=\sum_{i}\frac{d}{dt}|x_i-y_i|\\
			= & \sum_{i}\text{sign}(x_i-y_i)(\dot{x_i}-\dot{y_i})\\
			= & \sum_{i}\text{sign}(x_i-y_i)(g_i(x)-g_i(y))\\
			= &  \sum_{i}\text{sign}(x_i-y_i)(g_i(y+h)-g_i(y)),
		\end{aligned}																	
	\end{equation} 
	where $x=y+h$. From 
	\begin{equation*}
		\begin{aligned}
			g_i(y+h)-g_i(y)=& \int_0 ^1 \frac{dg_i(y+\tau h)}{d\tau} d\tau \\
			=&\int_0 ^1 \sum_{j}{\frac{\partial{g_i}}{\partial{y_j}}}h_j d\tau,
		\end{aligned}	
	\end{equation*}
	thus (\ref{eq1}) is equal to 
	\begin{equation*} \label{eq2}
		\int_0 ^1 \sum_{i}\text{sign}(h_i) \sum_{j}{\frac{\partial{g_i}}{\partial{z_j}}}h_j d\tau.
	\end{equation*}
	It holds that $\sum_{i}\text{sign}(h_i) \sum_{j}\frac{\partial g_i}{\partial y_j}h_j$ is equal to
	\begin{equation*}
		\begin{aligned}
			& \sum_{i}\left( \sum_{j \neq i} \frac{\partial g_i}{\partial y_j} h_j \text{sign}(h_i) + \frac{\partial g_i}{\partial y_i}|h_i|\right)\\
			\le &\sum_{i}\left( \sum_{j \neq i} \left|\frac{\partial g_i}{\partial y_j}\right| |h_j| + \frac{\partial g_i}{\partial y_i}|h_i|\right) \\
			= &\sum_{j} \sum_{i \neq j} \left|\frac{\partial g_i}{\partial y_j}\right| |h_j| + \sum_j \frac{\partial g_j}{\partial y_j}|h_j|\\
			= &\sum_j |h_j|\left(\sum_{i \neq j} \left|\frac{\partial g_i}{\partial y_j}\right| +\frac{\partial g_j}{\partial y_j} \right)\\
			\le & -||h||_1 c = -||x-y||_1 c.
		\end{aligned}
	\end{equation*}
	Plugging this in \eqref{eq1}, we get
	\begin{equation}
		\frac{d}{dt}\|x-y\|_1 \le -c ||x-y||_1,
	\end{equation}
	which implies \eqref{eq:l1_contract}. For point 2) we refer to \cite[Theorem 13]{jafarpour2020weak}.
\end{proof}
%	For every $z$, it holds:
%	\begin{equation}
%		\label{eq:jac}
%		\frac{\partial \dot{z}^p_i (\mathbf{z},\eta)}{\partial z^q_j}=\eta \tau^p\frac{\exp{(-\eta c_i^p(\mathbf{z}))}\sum_s \exp{(-\eta c_s^p(\mathbf{z}))}\big(\frac{\partial c_s^p(\mathbf{z})}{\partial z_j^{q}}-\frac{\partial c_i^p(\mathbf{z})}{\partial z_j^{q}}\big)}{(\sum_s \exp{(-\eta c_s^p(\mathbf{z}))})^2}-\delta_{ij} \delta_{pq},
%	\end{equation}
We can now proceed to the proof of Theorem \ref{thm:noise}.
Similarly to what done in the proof of Theorem \ref{thm:logit}-iii), we write logit$(\eta)$ in the form $\dot{z} = g(z,\eta) = G(z,\eta) - z$, and the Jacobian of $g$ as
\begin{equation*}
	J_{g,z}(z,\eta)=J_{G,z}(z,\eta)-\mb{I}.
\end{equation*}
Observe from \eqref{eq:jac3} that $J_{G,z}(z,\eta) \stackrel{\eta \to +\infty}{\longrightarrow} \mb 0$ for every $z \in \mc{Z}$ independently of the considered routing game. It thus follows that 
\begin{equation}
	\label{eq:Jg}
	\lim_{\eta \to +\infty}J_{g,z}(z,\eta)=-\mb{I} \quad \forall z \in \mc{Z}.
\end{equation}
%		$ for every $z \in \mc{Z}$ 
%	For every multigraph, $B(z,\eta)$ satisfies two properties:
%	\begin{enumerate}
%		\item \textcolor{red}{serve?} the diagonal is non-positive, because for every populations $p,q$ and routes $j,i$, 
%		$$\frac{\partial c_i^p(z)}{\partial z_j^q}\le \frac{\partial c_i^p(z)}{\partial z_i^p}; $$
%		\item $B(z,0)=\mb{0}%_{\RR\PP \times \RR\PP}
%		$ for every $z \in \mc{Z}$.
%		%\item $M(z,\eta)$ is bounded for every $z \in \mc{Z}$ and $\eta \ge 0$.
%	\end{enumerate}
%Thus, we have that for every multigraph and assignment of link cost functions it holds $J(z,0)=-\mb{I}$, i.e.,
With a slight abuse of notation, let from now on $J$ denote $J_{g,z}$.
From \eqref{eq:Jg}, it follows
\begin{align*}
	\lim_{\eta \to +\infty} \max_{j} \big(J_{jj}(z,\eta)+\sum_{i:i\neq j} |J_{ij}(z,\eta)|\big) = -1 \quad \forall z \in \mc{Z}.
\end{align*}
Since $J(z,\eta)$ is continuously differentiable in $\eta$, it follows that for every $k \in (0,1]$ there exists $\eta_k \ge 0$ such that for every $z \in \mc{Z}$ and $\eta \in [\eta_k,+\infty)$,
\begin{equation*}
	\label{eq:measure}
	\max_{j} \big(J_{jj}(z,\eta)+\sum_{i:i\neq j} |J_{ij}(z,\eta)|\big) \le -k \quad \forall \ z \in \mc{Z}.
\end{equation*}
%Using the fact that the domain of the dynamics is the the product of simplexes, which is compact and invariant under logit dynamics, t
Let $\eta_0$ be the largest $\eta>0$ such that
$$
\max_{z \in \mc{Z}} \max_{j} \left(J_{jj}(z,\eta)+\sum_{i\neq j} |J_{ij}(z,\eta)|\right) = 0.
$$
Thus, the existence of a globally exponentially stable fixed point for every $\eta \in (\underline{\eta},+\infty)$ follows from Proposition \ref{prp:contractivity} and from identifying $\underline{\eta}$ with $\eta_0$.
%The existence of a globally exponentially stable fixed point with rate $k$ in a right neighborhood of $\eta = 0$ thus follows from Proposition \ref{prp:contractivity}.

\end{document}